\documentclass[11pt]{llncs}

\usepackage{amssymb}
\usepackage{hyperref}
\usepackage{tikz}
\usepackage{graphicx}
\usepackage{enumitem}
\usepackage{subfigure}

\newcommand{\Z}{\mathbb Z}
\newcommand{\codespace}{{\Z_c}^{\!\! n}}
\newcommand{\complexityclass}[1]{\textup{\textbf{#1}}}
\newcommand{\computproblem}[1]{\textsc{#1}}

\newcommand{\SP}{\complexityclass{\#P}}
\newcommand{\NP}{\complexityclass{NP}}

\newcommand{\PSPACE}{\complexityclass{PSPACE}}
\newcommand{\ASP}{\complexityclass{ASP}}
\newcommand{\MSP}{\computproblem{MSP}}
\newcommand{\SMSP}{\computproblem{\#MSP}}
\newcommand{\MSPB}{\computproblem{MSP-Black}}
\newcommand{\SMSPB}{\computproblem{\#MSP-Black}}
\newcommand{\MSPW}{\computproblem{MSP-White}}
\newcommand{\SMSPW}{\computproblem{\#MSP-White}}
\newcommand{\MASTER}{\computproblem{Mastermind}}
\newcommand{\SAT}{\computproblem{SAT}}

\newcommand{\USAT}{\computproblem{USAT}}
\newcommand{\SMATCH}{\computproblem{\#Match}}
\newcommand{\sharpp}{\computproblem{\#}}
\newcommand*\circled[1]{\tikz[baseline=(char.base)]{\node[shape=circle,draw,inner sep=1pt](char){#1};}}

\begin{document}
\mainmatter

\title{Hardness of Mastermind}

\titlerunning{Hardness of Mastermind}

\author{Giovanni Viglietta}

\authorrunning{Giovanni Viglietta}

\tocauthor{Giovanni Viglietta}

\institute{University of Pisa, Italy,\\
\email{viglietta@gmail.com}
}
\maketitle

\begin{abstract}
Mastermind is a popular board game released in 1971, where a codemaker chooses a secret pattern of colored pegs, and a codebreaker has to guess it in several trials. After each attempt, the codebreaker gets a response from the codemaker containing some information on the number of correctly guessed pegs. The search space is thus reduced at each turn, and the game continues until the codebreaker is able to find the correct code, or runs out of trials.

In this paper we study several variations of \SMSP, the problem of computing the size of the search space resulting from a given (possibly fictitious) sequence of guesses and responses. Our main contribution is a proof of the \SP-completeness of \SMSP\ under parsimonious reductions, which settles an open problem posed by Stuckman and Zhang in 2005, concerning the complexity of deciding if the secret code is uniquely determined by the previous guesses and responses. Similarly, \SMSP\ stays \SP-complete under Turing reductions even with the promise that the search space has at least $k$ elements, for any constant $k$. (In a regular game of Mastermind, $k=1$.)

All our hardness results hold even in the most restrictive setting, in which there are only two available peg colors, and also if the codemaker's responses contain less information, for instance like in the so-called single-count (black peg) Mastermind variation.
\end{abstract}

\keywords{Mastermind; code-breaking; game; counting; search space.}

\section{Introduction}\label{s1}

\paragraph{\textbf{\emph{Mastermind at a glance.}}} \emph{Mastermind} is a code-breaking board game released in 1971, which sold over 50 million sets in 80 countries. The Israeli postmaster and telecommunication expert Mordecai Meirowitz is usually credited for inventing it in 1970, although an almost identical paper-and-pencil game called \emph{bulls and cows} predated Mastermind, perhaps by more than a century~\cite{wiki}.

\begin{figure}[ht]
\centering
\includegraphics[width=0.645\linewidth]{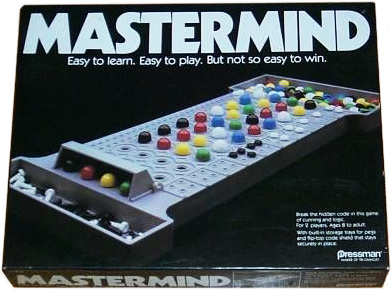}
\caption{A Mastermind box published by Pressman Toy Corporation in 1981, foreshadowing the game's computational hardness.}
\label{fig:1}
\end{figure}

The classic variation of the game is played between a \emph{codemaker}, who chooses a secret sequence of four colored pegs, and a \emph{codebreaker}, who tries to guess it in several attempts. There are six available colors, and the secret code may contain repeated colors. After each attempt, the codebreaker gets a \emph{rating} from the codemaker, consisting in the number of correctly placed pegs in the last guess, and the number of pegs that have the correct color but are misplaced. The rating does not tell which pegs are correct, but only their amount. These two numbers are communicated by the codemaker as a sequence of smaller black pegs and white pegs, respectively (see \hyperref[fig:1]{Figure~\ref*{fig:1}}, where the secret code is concealed behind a shield, and each guess is paired with its rating). If the codebreaker's last guess was wrong, he guesses again, and the game repeats until the secret code is found, or the codebreaker reaches his limit of ten trials. Ideally, the codebreaker plans his new guesses according to the information he collected from the previous guesses. \hyperref[tab:1]{Table~\ref*{tab:1}} depicts a complete game of Mastermind, where colors are encoded as numbers between zero and five, and the codebreaker finally guesses the code at his sixth attempt.

\begin{table}[h!b!p!]
\centering
\caption{A typical game of Mastermind.}
\begin{tabular}{|c|c|}
\hline
\multicolumn{2}{|c|}{Secret code: \circled{0} \circled{1} \circled{2} \circled{3}}\\
\hline
\hline
Guess & Rating\\
\hline
\circled{4} \circled{4} \circled{1} \circled{1} & $\circ$\\
\circled{3} \circled{2} \circled{2} \circled{4} & $\bullet$\,$\circ$\\
\circled{0} \circled{3} \circled{0} \circled{4} & $\bullet$\,$\circ$\\
\circled{5} \circled{5} \circled{3} \circled{4} & $\circ$\\
\circled{1} \circled{2} \circled{0} \circled{3} & $\bullet$\,$\circ$\,$\circ$\,$\circ$\\
\circled{0} \circled{1} \circled{2} \circled{3} & $\bullet$\,$\bullet$\,$\bullet$\,$\bullet$\\
\hline
\end{tabular}
\label{tab:1}
\end{table}

\paragraph{\textbf{\emph{Previous work.}}} Recently, Focardi and Luccio pointed out the unexpected relevance of Mastermind in real-life security issues, by showing how certain API-level bank frauds, aimed at disclosing user PINs, can be interpreted as an extended Mastermind game played between an insider and the bank's computers~\cite{luccio}. On the other hand, Goodrich suggested some applications to genetics of the Mastermind variation in which scores consist of black pegs only, called \emph{single-count (black peg) Mastermind}~\cite{goodrich}.

As a further generalization of the original game, we may consider \emph{$(n,c)$-Mastermind}, where the secret sequence consists of $n$ pegs, and there are $c$ available colors. Chv\'atal proved that the codebreaker can always determine the secret code in $(n,c)$-Mastermind after at most $2n\log_2 c + 4n + \lceil \frac c n\rceil$ guesses, each computable in polynomial time, via a simple divide-and-conquer strategy~\cite{chvatal}. This upper bound was later lowered by a constant factor in~\cite{chen}, while Goodrich also claimed to be able to lower it for single-count (black peg) Mastermind, hence using even less information~\cite{goodrich}. Unfortunately, after a careful inspection, Goodrich's method turns out to outperform Chv\'atal's several techniques given in~\cite{chvatal} asymptotically (as $n$ grows, and $c$ is a function of $n$) only if $n^{1-\varepsilon} < c < (3+\varepsilon)n\log_2 n$, for every $\varepsilon>0$.

However, despite being able to guess any secret code with an efficient strategy, the codebreaker may insist on really minimizing the number of trials, either in the worst case or on average. Knuth proposed a heuristic that exhaustively searches through all possible guesses and ratings, and greedily picks a guess that will minimize the number of eligible solutions, in the worst case~\cite{knuth}. This is practical and worst-case optimal for standard $(4,6)$-Mastermind, but infeasible and suboptimal for even slightly bigger instances. The size of the solution space is employed as an ideal quality indicator also in other heuristics, most notably those based on genetic algorithms~\cite{genetic}.

In order to approach the emerging complexity theoretic issues, Stuckman and Zhang introduced the \textsc{Mastermind Satisfiability Problem} (\MSP) for $(n,c)$-Mastermind, namely the problem of deciding if a given sequence of guesses and ratings has indeed a solution, and proved its \NP-completeness~\cite{zhang}. Similarly, Goodrich showed that also the analogous satisfiability problem for single-count (black peg) Mastermind is \NP-complete~\cite{goodrich}.

Interestingly, Stuckman and Zhang observed that the problem of detecting \MSP\ instances with a unique solution is Turing-reducible to the problem of producing an eligible solution. However, the determination of the exact complexity of the first problem is left open~\cite{zhang}.

\paragraph{\textbf{\emph{Our contribution.}}} In this paper we study \SMSP, the \emph{counting problem} associated with \MSP, i.e., the problem of computing the number of solutions that are compatible with a given set of guesses and ratings. We do this for standard $(n,c)$-Mastermind, as well as its single-count variation with only black peg ratings, and the analogous single-count variation with only white peg ratings, both in general and restricted to instances with a fixed number of colors $c$.

Our main theorem states that, in all the aforementioned variations of Mastermind, \SMSP\ is either trivially polynomial or \SP-complete under \emph{parsimonious reductions}. Capturing the true complexity of \SMSP\ is an improvement on previous results (refer to~\cite{goodrich,zhang}) because:
\begin{itemize}[topsep=2pt, partopsep=0pt]\itemsep2pt
\item[$\bullet$] Evaluating the size of the search space is a natural and recurring subproblem in several heuristics, whereas merely deciding if a set of guesses has a solution seems a more fictitious problem, especially because in a real game of Mastermind we already know that our previous guesses and ratings \emph{do} have a solution.
\item[$\bullet$] The reductions we give are parsimonious, hence they yield stronger versions of all the previously known \NP-completeness proofs for \MSP\ and its variations. Moreover, we obtain the same hardness results even for $(n,2)$-Mastermind, whereas all the previous reductions used unboundedly many colors (see \hyperref[cor:npc]{Corollary~\ref*{cor:npc}}).
\item[$\bullet$] Our main theorem enables simple proofs of a wealth of complexity-related corollaries, including the hardness of detecting unique solutions, which was left open in~\cite{zhang} (see \hyperref[cor:zhang]{Corollary~\ref*{cor:zhang}}).
\end{itemize}

\paragraph{\textbf{\emph{Paper structure.}}} In \hyperref[s2]{Section~\ref*{s2}} we define \SMSP\ and its variations. \hyperref[s3]{Section~\ref*{s3}} contains a statement and proof of our main result, \hyperref[thm:main]{Theorem~\ref*{thm:main}}, and an example of reduction. In \hyperref[s4]{Section~\ref*{s4}} we apply \hyperref[thm:main]{Theorem~\ref*{thm:main}} to several promise problems with different assumptions on the search space, and finally in \hyperref[s5]{Section~\ref*{s5}} we suggest some directions for further research.

\section{Definitions}\label{s2}

\paragraph{\textbf{\emph{Codes and ratings.}}} For $(n,c)$-Mastermind, let the set $\codespace$ be the \emph{code space}, whose elements are \emph{codes} of $n$ numbers ranging from $0$ to $c-1$. Following Chv\'atal, we define two \emph{metrics} on the code space~\cite{chvatal}. If $x=(x_1, \cdots, x_n)$ and $y=(y_1, \cdots, y_n)$ are two codes, let $\alpha(x,y)$ be the number of subscripts $i$ with $x_i=y_i$, and let $\beta(x,y)$ be the largest $\alpha(x,\tilde{y})$, with $\tilde{y}$ running through all the permutations of $y$. As observed in~\cite{zhang}, $n-\alpha(x,y)$ and $n-\beta(x,y)$ are indeed distance functions, respectively on $\codespace$ and $\codespace/S_n$ (i.e., the code space where codes are equivalent up to reordering of their elements).

Given a secret code $s\in \codespace$ chosen by the codemaker, we define the \emph{rating} of a guess $g\in \codespace$, for all the three variants of Mastermind we want to model.
\begin{itemize}
\item[$-$] For standard Mastermind, let $\rho(s,g)=\left(\alpha(s,g), \beta(s,g)-\alpha(s,g)\right)$.
\item[$-$] For single-count black peg Mastermind, let $\rho_b(s,g)=\alpha(s,g)$.
\item[$-$] For single-count white peg Mastermind, let $\rho_w(s,g)=\beta(s,g)$.
\end{itemize}
A guess is considered correct in single-count white peg $(n,c)$-Mastermind whenever its rating is $n$, therefore the secret code has to be guessed only up to reordering of the numbers. As a consequence, the codebreaker can always guess the code after $c-1$ attempts: He can determine the number of pegs of each color via monochromatic guesses, although this is not an optimal strategy when $c$ outgrows $n$. On the other hand, order does matter in both other variants of Mastermind, where the guess has to coincide with the secret code for the codebreaker to win.

\paragraph{\textbf{\emph{Satisfiability problems.}}} Next we define the \textsc{Mastermind Satisfiability Problem} for all three variants of Mastermind.
\begin{problem}
\MSP\ (respectively, \MSPB, \MSPW).\\
\textit{Input:} $(n,c,Q)$, where $Q$ is a finite set of queries of the form $(g,r)$, where $g\in \codespace$ and $r$ is a rating.\\
\textit{Output:} \textsc{Yes} if there exists a code $x\in \codespace$ such that $r=\rho(x,g)$ (respectively, $r=\rho_b(x,g)$, $r=\rho_w(x,g)$) for all $(g,r)\in Q$. \textsc{No} otherwise.
\end{problem}
\MSP\ and \MSPB\ are known to be \NP-complete problems~\cite{goodrich,zhang}. We shall see in \hyperref[cor:npc]{Corollary~\ref*{cor:npc}} how \MSPW\ is \NP-complete, as well.

Further, we may want to restrict our attention to instances of Mastermind with a fixed number of colors. Thus, for every constant $c$, let $(c)$-\MSP\ be the restriction of \MSP\ to problem instances with exactly $c$ colors (i.e., whose input is of the form $(\cdot, c, \cdot)$). Similarly, we define $(c)$-\MSPB\ and $(c)$-\MSPW.

\paragraph{\textbf{\emph{Counting problems.}}} All the above problems are clearly in \NP, thus it makes sense to consider their \emph{counting versions}, namely \SMSP, \SMSPB, \sharpp$(c)$-\MSP, and so on, which are all \SP\ problems~\cite{valiant1}. Basically, these problems ask for the size of the solution space after a number of guesses and ratings, i.e., the number of codes that are coherent with all the guesses and ratings given as input.

Recall that reductions among \SP\ problems that are based on oracles are called \emph{Turing reductions} and are denoted with $\leqslant_{\rm{T}}$, while the more specific reductions that map problem instances preserving the number of solutions are called \emph{parsimonious reductions}, and are denoted with $\leqslant_{\rm{pars}}$. Each type of reduction naturally leads to a different notion of \SP-completeness: For instance, \sharpp$2$-\SAT\ is \SP-complete under Turing reductions, while \sharpp$3$-SAT\ is \SP-complete under parsimonious reductions~\cite{papadimitriou}. Problems that are \SP-complete under parsimonious reductions are \emph{a fortiori} \NP-complete, while it is unknown whether all \NP-complete problems are \SP-complete, even under Turing reductions~\cite{dyer}.

\section{Counting Mastermind solutions}\label{s3}

Next we give a complete classification of the complexities of all the counting problems introduced in \hyperref[s2]{Section~\ref*{s2}}.

\begin{theorem}\label{thm:main}\hspace{0pt}
\begin{enumerate}[topsep=1pt, partopsep=0pt]
\item[\textbf{\emph{a)}}]\label{thm:maina} \SMSP, \SMSPB\ and \SMSPW\ are \SP-complete under parsimonious reductions.
\item[\textbf{\emph{b)}}]\label{thm:mainb} \sharpp$(c)$-\MSP\ and \sharpp$(c)$-\MSPB\ are \SP-complete under parsimonious reductions for every $c\geqslant 2$.
\item[\textbf{\emph{c)}}]\label{thm:mainc} \sharpp$(c)$-\MSPW\ is solvable in deterministic polynomial time for every $c\geqslant 1$.
\end{enumerate}
\end{theorem}
(Notice that \sharpp$(1)$-\MSP\ and \sharpp$(1)$-\MSPB\ are trivially solvable in deterministic linear time.)

\begin{lemma}\label{l1}
For every $c\geqslant 1$, \sharpp$(c)$-\MSPW\ is solvable in deterministic polynomial time.
\end{lemma}
\begin{proof}
In $\codespace/S_n$ there are only ${{n+c-1}\choose{c-1}}=\Theta (n^{c-1})$ possible codes to check against all the given queries, hence the whole process can be carried out in polynomial time, for any constant $c$.
\qed
\end{proof}

\begin{lemma}\label{l2}
For every $c\geqslant 1$,\\
\centerline{\sharpp$(c)$-\MSP\ $\leqslant_{\rm{pars}}$ \sharpp$(c+1)$-\MSP,}
\centerline{\sharpp$(c)$-\MSPB\ $\leqslant_{\rm{pars}}$ \sharpp$(c+1)$-\MSPB.}
\end{lemma}
\begin{proof}
Given the instance $(n,c,Q)$ of \sharpp$(c)$-\MSP\ (respectively, \sharpp$(c)$-\MSPB), we convert it into $\left(n,c+1,Q\cup\{(g,r)\}\right)$, where $g$ is a sequence of $n$ consecutive $c$'s, and $r=(0,0)$ (respectively, $r=0$). The new query $(g,r)$ implies that the new color $c$ does not occur in the secret code, hence the number of solutions is preserved and the reduction is indeed parsimonious.
\qed
\end{proof}

\begin{lemma}\label{l3}
\sharpp$3$-\SAT\ $\leqslant_{\rm{pars}}$ \SMSPW.
\end{lemma}
\begin{proof}
Given a 3-CNF Boolean formula $\varphi$ with $v$ variables and $m$ clauses, we map it into an instance of \MSPW\ $(n,c,Q)$. For each clause $C_i$ of $\varphi$, we add three fresh \emph{auxiliary variables} $a_i$, $b_i$, $c_i$. For each variable $x$ (including auxiliary variables), we define two colors $x$ and $\bar{x}$, representing the two possible truth assignments for $x$. We further add the \emph{mask color} $\ast$, thus getting $c=2v+6m+1$ colors in total. We let $n=v+3m$ (we may safely assume that $n\geqslant 5$), and we construct $Q$ as follows.
\begin{enumerate}[label=\arabic*)]
\item\label{l3s1} Add the query $\Big( (\ast, \ast, \ast, \cdots, \ast), 0 \Big)$.
\item\label{l3s2} For each variable $x$, add the query $\Big( (x, x, \bar{x}, \bar{x}, \ast, \ast, \ast, \cdots, \ast), 1 \Big)$.
\item\label{l3s3} For each clause $C_i=\{\ell_1 ,\ell_2, \ell_3\}$ (where each literal may be positive or negative), add the query $\Big( (\ell_1, \ell_2, \ell_3, a_i, b_i, \ast, \ast, \ast, \cdots, \ast), 3 \Big)$.
\item\label{l3s4} For each clause $C_i$, further add the query $\Big( (\overline{a_i}, b_i, c_i, \ast, \ast, \ast, \cdots, \ast ), 2 \Big)$.
\end{enumerate}
By~\hyperref[l3s1]{(1)}, the mask color does not occur in the secret code; by~\hyperref[l3s2]{(2)}, each variable occurs in the secret code exactly once, either as a positive or a negative literal. Moreover, by~\hyperref[l3s3]{(3)}, at least one literal from each clause must appear in the secret code. Depending on the exact number of literals from $C_i$ that appear in the code (either one, two or three), the queries in~\hyperref[l3s3]{(3)} and~\hyperref[l3s4]{(4)} always force the values of the auxiliary variables $a_i$, $b_i$ and $c_i$. (Notice that, without~\hyperref[l3s4]{(4)}, there would be two choices for $a_i$ and $b_i$, in case exactly two literals of $C_i$ appeared in the code.) As a consequence, the reduction is indeed parsimonious.
\qed
\end{proof}

\begin{lemma}\label{l4}
\sharpp$3$-\SAT\ $\leqslant_{\rm{pars}}$ \sharpp$(2)$-\MSPB.
\end{lemma}
\begin{proof}
We proceed along the lines of the proof of \hyperref[l3]{Lemma~\ref*{l3}}, with similar notation. We add the same auxiliary variables $a_i$, $b_i$, $c_i$ for each clause $C_i$, and we construct the instance of $(2)$-\MSPB\ $(2n,2,Q)$, where $n=v+3m$. This time we encode literals as \emph{positions} in the code: For each variable $x$, we allocate two specific positions $x$ and $\bar{x}$, so that $g_x=1$ (respectively, $g_{\bar{x}}=1$) in code $g=(g_1,\cdots,g_{2n})$ if and only if variable $x$ is assigned the value true (respectively, false). Notice that, in contrast with \hyperref[l3]{Lemma~\ref*{l3}}, we are not using a mask color here. $Q$ is constructed as follows.
\begin{enumerate}[label=\arabic*)]
\item\label{l4s1} Add the query $\Big( (0,0,0,\cdots,0), n \Big)$.
\item\label{l4s2} For each variable $x$, add the query $\left(g, n\right)$, where $g_j=1$ if and only if $j\in\{x, \bar{x}\}$.
\item\label{l4s3} For each clause $C_i=\{\ell_1 ,\ell_2, \ell_3\}$, add the query $\left(g, n+1\right)$, where $g_j=1$ if and only if $j\in\{\ell_1, \ell_2, \ell_3, a_i, b_i\}$. (Without loss of generality, we may assume that $\ell_1$, $\ell_2$ and $\ell_3$ are occurrences of three mutually distinct variables~\cite{yato}.)
\item\label{l4s4} For each clause $C_i$, further add the query $\left(g, n+1\right)$, where $g_j=1$ if and only if $j\in\{\overline{a_i}, b_i, c_i\}$.
\end{enumerate}
By~\hyperref[l4s1]{(1)}, every solution must contain $n$ times $0$ and $n$ times $1$, in some order. The semantics of~\hyperref[l4s2]{(2)}, \hyperref[l4s3]{(3)} and~\hyperref[l4s4]{(4)} is the same as that of the corresponding steps in \hyperref[l3]{Lemma~\ref*{l3}}, hence our construction yields the desired parsimonious reduction. Indeed, observe that, if altering $k$ bits of a binary code increases its rating by $r$, then exactly $\frac{k+r}{2}$ of those $k$ bits are set to the right value. In~\hyperref[l4s2]{(2)}, altering $k=2$ bits of the code in~\hyperref[l4s1]{(1)} increases its rating by $r=0$, hence exactly one of those bits has the right value, which means that $s_x\neq s_{\bar{x}}$ in any solution $s$. Similarly, in~\hyperref[l4s3]{(3)} (respectively,~\hyperref[l4s4]{(4)}), $k=5$ (respectively, $k=3$) and $r=1$, hence exactly three (respectively, two) of the bits set to $1$ are correct (cf.~the ratings in \hyperref[l3]{Lemma~\ref*{l3}}).
\qed
\end{proof}

\begin{lemma}\label{l5}
\sharpp$3$-\SAT\ $\leqslant_{\rm{pars}}$ \sharpp$(2)$-\MSP.
\end{lemma}
\begin{proof}
We replicate the construction given in the proof of \hyperref[l4]{Lemma~\ref*{l4}}, but we use the proper ratings: Recall that the ratings of \MSP\ are pairs of scores (black pegs and white pegs). The first score (black pegs) has the same meaning as in \MSPB, and we maintain these scores unchanged from the previous construction. By doing so, we already get the desired set of solutions, hence we merely have to show how to fill out the remaining scores (white pegs) without losing solutions.

Referring to the proof of \hyperref[l4]{Lemma~\ref*{l4}}, we change the rating in~\hyperref[l4s1]{(1)} from $n$ to $(n,0)$, because every $0$ in the guess is either correct at the correct place, or redundant.

The rating in~\hyperref[l4s2]{(2)} is changed from $n$ to $(n,2)$. Indeed, let $y$ be any other variable (distinct from $x$), so that $g_y=g_{\bar{y}}=0$. Then, exactly one between $g_y$ and $g_{\bar{y}}$ is a misplaced $0$, which can be switched with the misplaced $1$ from either $g_x$ or $g_{\bar{x}}$. All the other $0$'s in $g$ are either correct at the correct place, or redundant.

Similarly, the rating in~\hyperref[l4s3]{(3)} (respectively,~\hyperref[l4s4]{(4)}) changes from $n+1$ to $(n+1,4)$ (respectively, $(n+1,2)$). Indeed, exactly two (respectively, one) $1$'s are in a wrong position in $g$. If either $g_x=1$ or $g_{\bar{x}}=1$ is wrong, then both $g_x$ and $g_{\bar{x}}$ are wrong and of opposite colors, hence they can be switched. Once again, all the other $0$'s in $g$ are either correct at the correct place, or redundant.
\qed
\end{proof}

\begin{proof}[of {\hyperref[thm:main]{Theorem~\ref*{thm:main}}}]
All the claims easily follow from \hyperref[l1]{Lemma~\ref*{l1}}, \hyperref[l2]{Lemma~\ref*{l2}}, \hyperref[l3]{Lemma~\ref*{l3}}, \hyperref[l4]{Lemma~\ref*{l4}}, \hyperref[l5]{Lemma~\ref*{l5}}, and the \SP-completeness of \sharpp$3$-\SAT\ under parsimonious reductions~\cite{papadimitriou}.
\qed
\end{proof}

\paragraph{\textbf{\emph{Example.}}} As an illustration of \hyperref[l5]{Lemma~\ref*{l5}}, we show how the Boolean formula $(x\vee \neg y \vee z)\wedge(\neg x\vee y\vee w)\wedge(y\vee \neg z\vee \neg w)$ is translated into a set of queries for $(2)$-\MSP. For visual convenience, $0$'s and $1$'s are represented as white and black circles, respectively.

\vspace{-10pt}
\begin{table}[h!p!]
\centering
\begin{tabular}{|p{8pt}p{8pt}|p{8pt}p{8pt}|p{8pt}p{8pt}|p{8pt}p{8pt}|p{8pt}p{8pt}p{8pt}p{8pt}p{8pt}p{8pt}|p{8pt}p{8pt}p{8pt}p{8pt}p{8pt}p{8pt}|p{8pt}p{8pt}p{8pt}p{8pt}p{8pt}p{8pt}|c|}
\hline
\centering $x$ & \centering $\overline{x}$ & \centering $y$ & \centering $\overline{y}$ & \centering $z$ & \centering $\overline{z}$ & \centering $w$ & \centering $\overline{w}$ & \centering $a_1$ & \centering $\overline{a}_1$ & \centering $b_1$ & \centering $\overline{b}_1$ & \centering $c_1$ & \centering $\overline{c}_1$ & \centering $a_2$ & \centering $\overline{a}_2$ & \centering $b_2$ & \centering $\overline{b}_2$ & \centering $c_2$ & \centering $\overline{c}_2$ & \centering $a_3$ & \centering $\overline{a}_3$ & \centering $b_3$ & \centering $\overline{b}_3$ & \centering $c_3$ & \centering $\overline{c}_3$ & Rating\\
\hline
\centering $\circ$   & \centering $\circ$   & \centering $\circ$   & \centering $\circ$   & \centering $\circ$   & \centering $\circ$   & \centering $\circ$   & \centering $\circ$   & \centering $\circ$   & \centering $\circ$   & \centering $\circ$   & \centering $\circ$   & \centering $\circ$   & \centering $\circ$   & \centering $\circ$   & \centering $\circ$   & \centering $\circ$   & \centering $\circ$   & \centering $\circ$   & \centering $\circ$   & \centering $\circ$   & \centering $\circ$   & \centering $\circ$   & \centering $\circ$   & \centering $\circ$   & \centering $\circ$   & $(13,0)$\\
\hline
\centering $\bullet$ & \centering $\bullet$ & \centering $\circ$   & \centering $\circ$   & \centering $\circ$   & \centering $\circ$   & \centering $\circ$   & \centering $\circ$   & \centering $\circ$   & \centering $\circ$   & \centering $\circ$   & \centering $\circ$   & \centering $\circ$   & \centering $\circ$   & \centering $\circ$   & \centering $\circ$   & \centering $\circ$   & \centering $\circ$   & \centering $\circ$   & \centering $\circ$   & \centering $\circ$   & \centering $\circ$   & \centering $\circ$   & \centering $\circ$   & \centering $\circ$   & \centering $\circ$   & $(13,2)$\\
\centering $\circ$   & \centering $\circ$   & \centering $\bullet$ & \centering $\bullet$ & \centering $\circ$   & \centering $\circ$   & \centering $\circ$   & \centering $\circ$   & \centering $\circ$   & \centering $\circ$   & \centering $\circ$   & \centering $\circ$   & \centering $\circ$   & \centering $\circ$   & \centering $\circ$   & \centering $\circ$   & \centering $\circ$   & \centering $\circ$   & \centering $\circ$   & \centering $\circ$   & \centering $\circ$   & \centering $\circ$   & \centering $\circ$   & \centering $\circ$   & \centering $\circ$   & \centering $\circ$   & $(13,2)$\\
\centering $\circ$   & \centering $\circ$   & \centering $\circ$   & \centering $\circ$   & \centering $\bullet$ & \centering $\bullet$ & \centering $\circ$   & \centering $\circ$   & \centering $\circ$   & \centering $\circ$   & \centering $\circ$   & \centering $\circ$   & \centering $\circ$   & \centering $\circ$   & \centering $\circ$   & \centering $\circ$   & \centering $\circ$   & \centering $\circ$   & \centering $\circ$   & \centering $\circ$   & \centering $\circ$   & \centering $\circ$   & \centering $\circ$   & \centering $\circ$   & \centering $\circ$   & \centering $\circ$   & $(13,2)$\\
\centering $\circ$   & \centering $\circ$   & \centering $\circ$   & \centering $\circ$   & \centering $\circ$   & \centering $\circ$   & \centering $\bullet$ & \centering $\bullet$ & \centering $\circ$   & \centering $\circ$   & \centering $\circ$   & \centering $\circ$   & \centering $\circ$   & \centering $\circ$   & \centering $\circ$   & \centering $\circ$   & \centering $\circ$   & \centering $\circ$   & \centering $\circ$   & \centering $\circ$   & \centering $\circ$   & \centering $\circ$   & \centering $\circ$   & \centering $\circ$   & \centering $\circ$   & \centering $\circ$   & $(13,2)$\\
\centering $\circ$   & \centering $\circ$   & \centering $\circ$   & \centering $\circ$   & \centering $\circ$   & \centering $\circ$   & \centering $\circ$   & \centering $\circ$   & \centering $\bullet$ & \centering $\bullet$ & \centering $\circ$   & \centering $\circ$   & \centering $\circ$   & \centering $\circ$   & \centering $\circ$   & \centering $\circ$   & \centering $\circ$   & \centering $\circ$   & \centering $\circ$   & \centering $\circ$   & \centering $\circ$   & \centering $\circ$   & \centering $\circ$   & \centering $\circ$   & \centering $\circ$   & \centering $\circ$   & $(13,2)$\\
\centering $\circ$   & \centering $\circ$   & \centering $\circ$   & \centering $\circ$   & \centering $\circ$   & \centering $\circ$   & \centering $\circ$   & \centering $\circ$   & \centering $\circ$   & \centering $\circ$   & \centering $\bullet$ & \centering $\bullet$ & \centering $\circ$   & \centering $\circ$   & \centering $\circ$   & \centering $\circ$   & \centering $\circ$   & \centering $\circ$   & \centering $\circ$   & \centering $\circ$   & \centering $\circ$   & \centering $\circ$   & \centering $\circ$   & \centering $\circ$   & \centering $\circ$   & \centering $\circ$   & $(13,2)$\\
\centering $\circ$   & \centering $\circ$   & \centering $\circ$   & \centering $\circ$   & \centering $\circ$   & \centering $\circ$   & \centering $\circ$   & \centering $\circ$   & \centering $\circ$   & \centering $\circ$   & \centering $\circ$   & \centering $\circ$   & \centering $\bullet$ & \centering $\bullet$ & \centering $\circ$   & \centering $\circ$   & \centering $\circ$   & \centering $\circ$   & \centering $\circ$   & \centering $\circ$   & \centering $\circ$   & \centering $\circ$   & \centering $\circ$   & \centering $\circ$   & \centering $\circ$   & \centering $\circ$   & $(13,2)$\\
\centering $\circ$   & \centering $\circ$   & \centering $\circ$   & \centering $\circ$   & \centering $\circ$   & \centering $\circ$   & \centering $\circ$   & \centering $\circ$   & \centering $\circ$   & \centering $\circ$   & \centering $\circ$   & \centering $\circ$   & \centering $\circ$   & \centering $\circ$   & \centering $\bullet$ & \centering $\bullet$ & \centering $\circ$   & \centering $\circ$   & \centering $\circ$   & \centering $\circ$   & \centering $\circ$   & \centering $\circ$   & \centering $\circ$   & \centering $\circ$   & \centering $\circ$   & \centering $\circ$   & $(13,2)$\\
\centering $\circ$   & \centering $\circ$   & \centering $\circ$   & \centering $\circ$   & \centering $\circ$   & \centering $\circ$   & \centering $\circ$   & \centering $\circ$   & \centering $\circ$   & \centering $\circ$   & \centering $\circ$   & \centering $\circ$   & \centering $\circ$   & \centering $\circ$   & \centering $\circ$   & \centering $\circ$   & \centering $\bullet$ & \centering $\bullet$ & \centering $\circ$   & \centering $\circ$   & \centering $\circ$   & \centering $\circ$   & \centering $\circ$   & \centering $\circ$   & \centering $\circ$   & \centering $\circ$   & $(13,2)$\\
\centering $\circ$   & \centering $\circ$   & \centering $\circ$   & \centering $\circ$   & \centering $\circ$   & \centering $\circ$   & \centering $\circ$   & \centering $\circ$   & \centering $\circ$   & \centering $\circ$   & \centering $\circ$   & \centering $\circ$   & \centering $\circ$   & \centering $\circ$   & \centering $\circ$   & \centering $\circ$   & \centering $\circ$   & \centering $\circ$   & \centering $\bullet$ & \centering $\bullet$ & \centering $\circ$   & \centering $\circ$   & \centering $\circ$   & \centering $\circ$   & \centering $\circ$   & \centering $\circ$   & $(13,2)$\\
\centering $\circ$   & \centering $\circ$   & \centering $\circ$   & \centering $\circ$   & \centering $\circ$   & \centering $\circ$   & \centering $\circ$   & \centering $\circ$   & \centering $\circ$   & \centering $\circ$   & \centering $\circ$   & \centering $\circ$   & \centering $\circ$   & \centering $\circ$   & \centering $\circ$   & \centering $\circ$   & \centering $\circ$   & \centering $\circ$   & \centering $\circ$   & \centering $\circ$   & \centering $\bullet$ & \centering $\bullet$ & \centering $\circ$   & \centering $\circ$   & \centering $\circ$   & \centering $\circ$   & $(13,2)$\\
\centering $\circ$   & \centering $\circ$   & \centering $\circ$   & \centering $\circ$   & \centering $\circ$   & \centering $\circ$   & \centering $\circ$   & \centering $\circ$   & \centering $\circ$   & \centering $\circ$   & \centering $\circ$   & \centering $\circ$   & \centering $\circ$   & \centering $\circ$   & \centering $\circ$   & \centering $\circ$   & \centering $\circ$   & \centering $\circ$   & \centering $\circ$   & \centering $\circ$   & \centering $\circ$   & \centering $\circ$   & \centering $\bullet$ & \centering $\bullet$ & \centering $\circ$   & \centering $\circ$   & $(13,2)$\\
\centering $\circ$   & \centering $\circ$   & \centering $\circ$   & \centering $\circ$   & \centering $\circ$   & \centering $\circ$   & \centering $\circ$   & \centering $\circ$   & \centering $\circ$   & \centering $\circ$   & \centering $\circ$   & \centering $\circ$   & \centering $\circ$   & \centering $\circ$   & \centering $\circ$   & \centering $\circ$   & \centering $\circ$   & \centering $\circ$   & \centering $\circ$   & \centering $\circ$   & \centering $\circ$   & \centering $\circ$   & \centering $\circ$   & \centering $\circ$   & \centering $\bullet$ & \centering $\bullet$ & $(13,2)$\\
\hline
\centering $\bullet$ & \centering $\circ$   & \centering $\circ$   & \centering $\bullet$ & \centering $\bullet$ & \centering $\circ$   & \centering $\circ$   & \centering $\circ$   & \centering $\bullet$ & \centering $\circ$   & \centering $\bullet$ & \centering $\circ$   & \centering $\circ$   & \centering $\circ$   & \centering $\circ$   & \centering $\circ$   & \centering $\circ$   & \centering $\circ$   & \centering $\circ$   & \centering $\circ$   & \centering $\circ$   & \centering $\circ$   & \centering $\circ$   & \centering $\circ$   & \centering $\circ$   & \centering $\circ$   & $(14,4)$\\
\centering $\circ$   & \centering $\bullet$ & \centering $\bullet$ & \centering $\circ$   & \centering $\circ$   & \centering $\circ$   & \centering $\bullet$ & \centering $\circ$   & \centering $\circ$   & \centering $\circ$   & \centering $\circ$   & \centering $\circ$   & \centering $\circ$   & \centering $\circ$   & \centering $\bullet$ & \centering $\circ$   & \centering $\bullet$ & \centering $\circ$   & \centering $\circ$   & \centering $\circ$   & \centering $\circ$   & \centering $\circ$   & \centering $\circ$   & \centering $\circ$   & \centering $\circ$   & \centering $\circ$   & $(14,4)$\\
\centering $\circ$   & \centering $\circ$   & \centering $\bullet$ & \centering $\circ$   & \centering $\circ$   & \centering $\bullet$ & \centering $\circ$   & \centering $\bullet$ & \centering $\circ$   & \centering $\circ$   & \centering $\circ$   & \centering $\circ$   & \centering $\circ$   & \centering $\circ$   & \centering $\circ$   & \centering $\circ$   & \centering $\circ$   & \centering $\circ$   & \centering $\circ$   & \centering $\circ$   & \centering $\bullet$ & \centering $\circ$   & \centering $\bullet$ & \centering $\circ$   & \centering $\circ$   & \centering $\circ$   & $(14,4)$\\
\hline
\centering $\circ$   & \centering $\circ$   & \centering $\circ$   & \centering $\circ$   & \centering $\circ$   & \centering $\circ$   & \centering $\circ$   & \centering $\circ$   & \centering $\circ$   & \centering $\bullet$ & \centering $\bullet$ & \centering $\circ$   & \centering $\bullet$ & \centering $\circ$   & \centering $\circ$   & \centering $\circ$   & \centering $\circ$   & \centering $\circ$   & \centering $\circ$   & \centering $\circ$   & \centering $\circ$   & \centering $\circ$   & \centering $\circ$   & \centering $\circ$   & \centering $\circ$   & \centering $\circ$   & $(14,2)$\\
\centering $\circ$   & \centering $\circ$   & \centering $\circ$   & \centering $\circ$   & \centering $\circ$   & \centering $\circ$   & \centering $\circ$   & \centering $\circ$   & \centering $\circ$   & \centering $\circ$   & \centering $\circ$   & \centering $\circ$   & \centering $\circ$   & \centering $\circ$   & \centering $\circ$   & \centering $\bullet$ & \centering $\bullet$ & \centering $\circ$   & \centering $\bullet$ & \centering $\circ$   & \centering $\circ$   & \centering $\circ$   & \centering $\circ$   & \centering $\circ$   & \centering $\circ$   & \centering $\circ$   & $(14,2)$\\
\centering $\circ$   & \centering $\circ$   & \centering $\circ$   & \centering $\circ$   & \centering $\circ$   & \centering $\circ$   & \centering $\circ$   & \centering $\circ$   & \centering $\circ$   & \centering $\circ$   & \centering $\circ$   & \centering $\circ$   & \centering $\circ$   & \centering $\circ$   & \centering $\circ$   & \centering $\circ$   & \centering $\circ$   & \centering $\circ$   & \centering $\circ$   & \centering $\circ$   & \centering $\circ$   & \centering $\bullet$ & \centering $\bullet$ & \centering $\circ$   & \centering $\bullet$ & \centering $\circ$   & $(14,2)$\\
\hline
\centering $x$ & \centering $\overline{x}$ & \centering $y$ & \centering $\overline{y}$ & \centering $z$ & \centering $\overline{z}$ & \centering $w$ & \centering $\overline{w}$ & \centering $a_1$ & \centering $\overline{a}_1$ & \centering $b_1$ & \centering $\overline{b}_1$ & \centering $c_1$ & \centering $\overline{c}_1$ & \centering $a_2$ & \centering $\overline{a}_2$ & \centering $b_2$ & \centering $\overline{b}_2$ & \centering $c_2$ & \centering $\overline{c}_2$ & \centering $a_3$ & \centering $\overline{a}_3$ & \centering $b_3$ & \centering $\overline{b}_3$ & \centering $c_3$ & \centering $\overline{c}_3$ & Rating\\
\hline
\end{tabular}
\end{table}

\vspace{-10pt}
The solutions to both problems are exactly ten, and are listed below.
\vspace{-15pt}
\begin{table}[h!p!]
\centering
\subtable{
\begin{tabular}{|p{8pt}p{8pt}p{8pt}p{8pt}|}
\hline
\centering $x$ & \centering $y$ & \centering $z$ & {\centering $w$}\\
\hline
\centering T & \centering T & \centering T & {\centering T}\\
\centering T & \centering T & \centering T & {\centering F}\\
\centering T & \centering T & \centering F & {\centering T}\\
\centering T & \centering T & \centering F & {\centering F}\\
\centering T & \centering F & \centering F & {\centering T}\\
\centering F & \centering T & \centering T & {\centering T}\\
\centering F & \centering T & \centering T & {\centering F}\\
\centering F & \centering F & \centering T & {\centering F}\\
\centering F & \centering F & \centering F & {\centering T}\\
\centering F & \centering F & \centering F & {\centering F}\\
\hline
\end{tabular}
}\hspace{5pt}
\subtable{
\begin{tabular}{|p{8pt}p{8pt}|p{8pt}p{8pt}|p{8pt}p{8pt}|p{8pt}p{8pt}|p{8pt}p{8pt}p{8pt}p{8pt}p{8pt}p{8pt}|p{8pt}p{8pt}p{8pt}p{8pt}p{8pt}p{8pt}|p{8pt}p{8pt}p{8pt}p{8pt}p{8pt}p{8pt}|}
\hline
\centering $x$ & \centering $\overline{x}$ & \centering $y$ & \centering $\overline{y}$ & \centering $z$ & \centering $\overline{z}$ & \centering $w$ & \centering $\overline{w}$ & \centering $a_1$ & \centering $\overline{a}_1$ & \centering $b_1$ & \centering $\overline{b}_1$ & \centering $c_1$ & \centering $\overline{c}_1$ & \centering $a_2$ & \centering $\overline{a}_2$ & \centering $b_2$ & \centering $\overline{b}_2$ & \centering $c_2$ & \centering $\overline{c}_2$ & \centering $a_3$ & \centering $\overline{a}_3$ & \centering $b_3$ & \centering $\overline{b}_3$ & \centering $c_3$ & {\centering $\overline{c}_3$}\\
\hline
\centering $\bullet$ & \centering $\circ$   & \centering $\bullet$ & \centering $\circ$   & \centering $\bullet$ & \centering $\circ$   & \centering $\bullet$ & \centering $\circ$   & \centering $\circ$   & \centering $\bullet$ & \centering $\bullet$ & \centering $\circ$   & \centering $\circ$   & \centering $\bullet$ & \centering $\circ$   & \centering $\bullet$ & \centering $\bullet$ & \centering $\circ$   & \centering $\circ$   & \centering $\bullet$ & \centering $\bullet$ & \centering $\circ$ & \centering $\bullet$ & \centering $\circ$   & \centering $\bullet$ & {\centering $\circ$} \\
\centering $\bullet$ & \centering $\circ$   & \centering $\bullet$ & \centering $\circ$   & \centering $\bullet$ & \centering $\circ$   & \centering $\circ$   & \centering $\bullet$ & \centering $\circ$   & \centering $\bullet$ & \centering $\bullet$ & \centering $\circ$   & \centering $\circ$   & \centering $\bullet$ & \centering $\bullet$ & \centering $\circ$   & \centering $\bullet$ & \centering $\circ$   & \centering $\bullet$ & \centering $\circ$   & \centering $\circ$   & \centering $\bullet$ & \centering $\bullet$ & \centering $\circ$   & \centering $\circ$   & {\centering $\bullet$} \\
\centering $\bullet$ & \centering $\circ$   & \centering $\bullet$ & \centering $\circ$   & \centering $\circ$   & \centering $\bullet$ & \centering $\bullet$ & \centering $\circ$   & \centering $\bullet$ & \centering $\circ$   & \centering $\bullet$ & \centering $\circ$   & \centering $\bullet$ & \centering $\circ$   & \centering $\circ$   & \centering $\bullet$ & \centering $\bullet$ & \centering $\circ$   & \centering $\circ$   & \centering $\bullet$ & \centering $\circ$   & \centering $\bullet$ & \centering $\bullet$ & \centering $\circ$   & \centering $\circ$   & {\centering $\bullet$} \\
\centering $\bullet$ & \centering $\circ$   & \centering $\bullet$ & \centering $\circ$   & \centering $\circ$   & \centering $\bullet$ & \centering $\circ$   & \centering $\bullet$ & \centering $\bullet$ & \centering $\circ$   & \centering $\bullet$ & \centering $\circ$   & \centering $\bullet$ & \centering $\circ$   & \centering $\bullet$ & \centering $\circ$   & \centering $\bullet$ & \centering $\circ$   & \centering $\bullet$ & \centering $\circ$   & \centering $\circ$   & \centering $\bullet$ & \centering $\circ$   & \centering $\bullet$ & \centering $\bullet$ & {\centering $\circ$}   \\
\centering $\bullet$ & \centering $\circ$   & \centering $\circ$   & \centering $\bullet$ & \centering $\circ$   & \centering $\bullet$ & \centering $\bullet$ & \centering $\circ$   & \centering $\circ$   & \centering $\bullet$ & \centering $\bullet$ & \centering $\circ$   & \centering $\circ$   & \centering $\bullet$ & \centering $\bullet$ & \centering $\circ$   & \centering $\bullet$ & \centering $\circ$   & \centering $\bullet$ & \centering $\circ$   & \centering $\bullet$ & \centering $\circ$   & \centering $\bullet$ & \centering $\circ$   & \centering $\bullet$ & {\centering $\circ$}   \\
\centering $\circ$   & \centering $\bullet$ & \centering $\bullet$ & \centering $\circ$   & \centering $\bullet$ & \centering $\circ$   & \centering $\bullet$ & \centering $\circ$   & \centering $\bullet$ & \centering $\circ$   & \centering $\bullet$ & \centering $\circ$   & \centering $\bullet$ & \centering $\circ$   & \centering $\circ$   & \centering $\bullet$ & \centering $\circ$   & \centering $\bullet$ & \centering $\bullet$ & \centering $\circ$   & \centering $\bullet$ & \centering $\circ$   & \centering $\bullet$ & \centering $\circ$   & \centering $\bullet$ & {\centering $\circ$}   \\
\centering $\circ$   & \centering $\bullet$ & \centering $\bullet$ & \centering $\circ$   & \centering $\bullet$ & \centering $\circ$   & \centering $\circ$   & \centering $\bullet$ & \centering $\bullet$ & \centering $\circ$   & \centering $\bullet$ & \centering $\circ$   & \centering $\bullet$ & \centering $\circ$   & \centering $\circ$   & \centering $\bullet$ & \centering $\bullet$ & \centering $\circ$   & \centering $\circ$   & \centering $\bullet$ & \centering $\circ$   & \centering $\bullet$ & \centering $\bullet$ & \centering $\circ$   & \centering $\circ$   & {\centering $\bullet$} \\
\centering $\circ$   & \centering $\bullet$ & \centering $\circ$   & \centering $\bullet$ & \centering $\bullet$ & \centering $\circ$   & \centering $\circ$   & \centering $\bullet$ & \centering $\circ$   & \centering $\bullet$ & \centering $\bullet$ & \centering $\circ$   & \centering $\circ$   & \centering $\bullet$ & \centering $\bullet$ & \centering $\circ$   & \centering $\bullet$ & \centering $\circ$   & \centering $\bullet$ & \centering $\circ$   & \centering $\bullet$ & \centering $\circ$   & \centering $\bullet$ & \centering $\circ$   & \centering $\bullet$ & {\centering $\circ$}   \\
\centering $\circ$   & \centering $\bullet$ & \centering $\circ$   & \centering $\bullet$ & \centering $\circ$   & \centering $\bullet$ & \centering $\bullet$ & \centering $\circ$   & \centering $\bullet$ & \centering $\circ$   & \centering $\bullet$ & \centering $\circ$   & \centering $\bullet$ & \centering $\circ$   & \centering $\circ$   & \centering $\bullet$ & \centering $\bullet$ & \centering $\circ$   & \centering $\circ$   & \centering $\bullet$ & \centering $\bullet$ & \centering $\circ$   & \centering $\bullet$ & \centering $\circ$   & \centering $\bullet$ & {\centering $\circ$}   \\
\centering $\circ$   & \centering $\bullet$ & \centering $\circ$   & \centering $\bullet$ & \centering $\circ$   & \centering $\bullet$ & \centering $\circ$   & \centering $\bullet$ & \centering $\bullet$ & \centering $\circ$   & \centering $\bullet$ & \centering $\circ$   & \centering $\bullet$ & \centering $\circ$   & \centering $\bullet$ & \centering $\circ$   & \centering $\bullet$ & \centering $\circ$   & \centering $\bullet$ & \centering $\circ$   & \centering $\circ$   & \centering $\bullet$ & \centering $\bullet$ & \centering $\circ$   & \centering $\circ$   & {\centering $\bullet$} \\
\hline
\end{tabular}
}
\end{table}

We remark that, in order to determine the values of the auxiliary variables $a_i$, $b_i$ and $c_i$ when a solution to the Boolean satisfiability problem is given, it is sufficient to check how many literals of $C_i$ are satisfied. $a_i$ is true if and only if exactly one literal is satisfied, $b_i$ is false if and only if all three literals are satisfied, and $c_i$ is true if and only if $a_i=b_i$.

\section{Related results}\label{s4}

We describe some applications of \hyperref[thm:main]{Theorem~\ref*{thm:main}} to several complexity problems.

\begin{corollary}\label{cor:npc}
$(2)$-\MSP, $(2)$-\MSPB\ and \MSPW\ are \NP-complete.
\end{corollary}
\begin{proof}
Parsimonious reductions among \SP\ problems are \emph{a fortiori} Karp reductions among the corresponding \NP\ problems.
\qed
\end{proof}

So far, we made no assumptions on the queries in our problem instances, which leads to a more general but somewhat fictitious theory. Since in a real game of Mastermind the codebreaker's queries are guaranteed to have at least a solution (i.e., the secret code chosen by the codemaker), more often than not the codebreaker is in a position to exploit this information to his advantage. However, we show that such information does not make counting problems substantially easier.

\begin{corollary}
\sharpp$(2)$-\MSP, \sharpp$(2)$-\MSPB\ and \SMSPW, with the promise that the number of solutions is at least $k$, are all \SP-complete problems under Turing reductions, for every $k\geqslant 1$.
\end{corollary}
\begin{proof}
Let \SMATCH\ be the problem of counting the matchings of any size in a given graph, which is known to be \SP-complete under Turing reductions~\cite{valiant2}. Let ${\rm{\Pi}}_k$ be the problem \sharpp$(2)$-\MSP\ (respectively, \sharpp$(2)$-\MSPB, \SMSPW) restricted to instances with at least $k$ solutions, and let us show that \SMATCH\ $\leqslant_{\rm{T}}$ ${\rm{\Pi}}_k$. Given a graph $G$, if it has fewer than $k$ edges, we can count all the matchings in linear time. Otherwise, there must be at least $k$ matchings (each edge $e$ yields at least the matching $\{e\}$), so we parsimoniously map $G$ into an instance of ${\rm{\Pi}}_k$ via \hyperref[thm:main]{Theorem~\ref*{thm:main}}, we call an oracle for ${\rm{\Pi}}_k$, and output its answer.
\qed
\end{proof}

The following result, for $k=1$, settles an issue concerning the determination of \MSP\ instances with unique solution, which was left unsolved in~\cite{zhang}. We actually prove more: Even if a solution is given as input, it is hard to determine if it is unique. Therefore, not only solving Mastermind puzzles is hard, but \emph{designing} puzzles around a solution is also hard.

\begin{corollary}\label{cor:zhang}
For every $k\geqslant 1$, the problem of deciding if an instance of $(2)$-\MSP, $(2)$-\MSPB\ or \MSPW\ has strictly more than $k$ solutions is \NP-complete, even if $k$ solutions are explicitly given as input.
\end{corollary}
\begin{proof}
Not only do the parsimonious reductions given in \hyperref[thm:main]{Theorem~\ref*{thm:main}} preserve the number of solutions, but they actually yield an explicit polynomial-time computable transformation of solutions (cf.~the remark at the end of \hyperref[s3]{Section~\ref*{s3}}). Hence, the involved \SP-complete problems are also \ASP-complete as function problems, and their decision $k$-\ASP\ counterparts are accordingly \NP-complete~\cite{yato}.
\qed
\end{proof}

Remarkably, even if the codebreaker somehow knows that his previous queries are sufficient to uniquely determine the solution, he still has a hard time finding it.

\begin{corollary}\label{cor:unique}
The promise problem of finding the solution to an instance of $(2)$-\MSP, $(2)$-\MSPB\ or \MSPW, when the solution itself is known to be unique, is \NP-hard under randomized Turing reductions.
\end{corollary}
\begin{proof}
It is known that \SAT\ $\leqslant_{\bf RP}$ \USAT, where \USAT\ is the promise version of \SAT\ whose input formulas are known to have either zero or one satisfying assignments~\cite{vazirani}. Let $f$ be the composition of this reduction with the parsimonious one from Boolean formulas to instances of $(2)$-\MSP\ (respectively, $(2)$-\MSPB, \MSPW) given by \hyperref[thm:main]{Theorem~\ref*{thm:main}}. Our Turing reduction proceeds as follows: Given a Boolean formula $\varphi$, compute $f(\varphi)$ and submit it to an oracle that finds a correct solution $s$ of $(2)$-\MSP\ (respectively, $(2)$-\MSPB, \MSPW) when it is unique. Then output \textsc{Yes} if and only if $s$ is indeed a solution of $f(\varphi)$, which can be checked in polynomial time.
\qed
\end{proof}

\section{Further research}\label{s5}
In \hyperref[l1]{Lemma~\ref*{l1}} we showed that \sharpp$(c)$-\MSPW\ is solvable in polynomial time when $c$ is a constant, while in \hyperref[l3]{Lemma~\ref*{l3}} we proved that it becomes \SP-complete when $c=2n+1$. By making the code polynomially longer and filling the extra space with a fresh color, we can easily prove that also \sharpp$\left(\Theta(\sqrt[k]n)\right)$-\MSPW\ is \SP-complete, for every constant $k$. An obvious question arises: What is the lowest order of growth of $c(n)$ such that \sharpp$(c(n))$-\MSPW\ is \SP-complete?

We observed that \SMSP\ is a subproblem of several heuristics aimed at optimally guessing the secret code, but is \SMSP\ really inherent in the game? Perhaps the hardness of Mastermind is not captured by \SMSP\ or even \MSP, and there are cleverer, yet unknown, ways to play.

\begin{problem}
\MASTER.\\
\textit{Input:} $(n,c,Q,k)$, where $(n,c,Q)$ is an instance of \MSP, and $k\geqslant 0$.\\
\textit{Output:} \textsc{Yes} if the codebreaker has a strategy to guess the secret code in at most $k$ attempts, using information from $Q$. \textsc{No} otherwise.
\end{problem}
Notice that \MASTER\ belongs to \PSPACE, due to the polynomial upper bound on the length of the optimal strategy given by Chv\'atal~\cite{chvatal}. Our question is whether \MASTER\ is \PSPACE-complete.

To make the game more fun to play for the codemaker, whose role is otherwise too passive, we could let him change the secret code at every turn, coherently with the ratings of the previous guesses of the codebreaker. As a result, nothing changes for the codebreaker, except that he may perceive to be quite unlucky with his guesses, but the codemaker's game becomes rather interesting: By \hyperref[cor:zhang]{Corollary~\ref*{cor:zhang}}, even deciding if he has a non-trivial move is \NP-complete, but he can potentially force the codebreaker to always work in the worst-case scenario, and make him pay for his mistakes. We call this variation \emph{adaptive Mastermind}.

\bibliographystyle{plain}

\end{document}